\documentclass{ifacconf}

\usepackage{graphicx}      
\usepackage{natbib}        
\usepackage{amsmath}
\usepackage{amsfonts}
\usepackage{comment}
\usepackage{float}
\usepackage{subcaption}
\usepackage{siunitx}

\newtheorem{theorem}{Theorem}
\newtheorem{proposition}{Proposition}
\newtheorem{corollary}{Corollary}

\newtheorem{lemma}{Lemma}
\newtheorem{assumption}{Assumption}

\begin{document}
\begin{frontmatter}
\title{Tracking the Effective Surface Area of Non-Convex Satellites} 

\thanks[footnoteinfo]{
This work was partially funded by the Research Council of Norway through the project 335832 QBDebris: A CubeSat formation for space debris characterization, and project 353062 Automatic super agile small-satellite operations}


\author[NTNU,SINTEF]{Lauritz R. Fosso} 
\author[UIT]{Raymond Kristiansen} 
\author[NTNU]{Jan Tommy Gravdahl}
\author[SINTEF]{Sveinung J. Ohrem}
\author[UIT]{Alessio Bocci}

\address[NTNU]{Department of Engineering Cybernetics, Norwegian University of Science and Technology (NTNU), Trondheim, Norway}
\address[SINTEF]{Department of Energy and Transport, SINTEF Ocean, Trondheim, Norway}
\address[UIT]{Department of Electrical Engineering, The Arctic University of Norway (UiT), Narvik, Norway}

\begin{abstract}                

This paper presents a novel framework to track the effective surface area of non-convex satellites, enabling the use of aerodynamic drag in low Earth orbit for orbital control. The proposed framework enables the satellite to track the effective surface area while simultaneously performing other maneuvers. We introduce this framework through a backstepping control algorithm, and exemplify its advantages with an extension, to simultaneously maximize solar panel exposure. The equilibria of the closed-loop systems are shown to be asymptotically stable, and simulation results confirm the effectiveness of the proposed framework.


\end{abstract}

\begin{keyword}
Satellite control, Non-convex satellite, Effective surface area, Cross-sectional area, Sun pointing, Switched system, Sliding surface, Backstepping 
\end{keyword}

\end{frontmatter}

\section{Introduction}
\subsection{Drag Control} \label{ch:Intro:sec:drag}
The increasing popularity of CubeSats has made control of these satellites a rapidly growing research topic. CubeSats are often equipped with reaction wheels or magnetorquers for attitude control. However, these satellites must adhere to the CubeSat specification, which sets restrictions on the types of thrusters that can be used, and as a result, only a limited number of CubeSats have been launched with thrusters \citep{lemmer_propulsion_2017}. An alternative method for orbital control is to leverage the effects of drag in LEO (Low Earth Orbit). The atmospheric drag is given by 
\begin{equation}
    \boldsymbol{F}_{drag} = - \frac{1}{2} \rho C S_{\mathrm{eff}} \Vert \boldsymbol{v} \Vert \boldsymbol{v}
\end{equation}
where $\rho$ is the atmospheric density, $C$ is the drag coefficient, $\boldsymbol{v}$ is the relative velocity of the satellite with respect to the atmosphere, and $S_{\mathrm{eff}}$ is the effective surface area \citep{markley_fundamentals_2014}. $S_{\mathrm{eff}}$ is a function of the attitude. Thus, the drag in LEO presents itself as a coupling between attitude dynamics and orbital dynamics. This enables the use of actuators traditionally used for attitude control for orbit manipulation. \\

The idea of exploiting drag effects in LEO for satellite control has gained traction in recent years, but has its origins in the space race. According to \cite{pande_optimal_1979}, the first implementation of this idea was reported by Sarychev in 1968, successfully demonstrating pitch-control on the KOSMOS-149 satellite using aerodynamic effects. This sparked an interest, and several articles were subsequently published on the topic. \cite{modi_optimized_1973} proposed a  controller that used aerodynamic torques to stabilize a cylindrical satellite equipped with flaps to an arbitrary orientation. \cite{pande_optimal_1979} extended this idea to spin-stabilized satellites, proposing a linear-quadratic regulator to control the two remaining degrees of freedom. \\

In later years focus shifted towards using drag for the purposes of formation-keeping and rendezvous operations. \cite{leonard_orbital_1989} first assessed the feasibility of using drag forces and torques as a means to control the formation-keeping of two satellites. This was realized using simple on-off control, that expand or retract drag plates. Leonard's work has subsequently been further expanded upon by numerous researchers. \cite{perez_differential_2013} proposed implementing a more sophisticated adaptive controller, but kept the design limitation that drag plates are either fully deployed or retracted. \cite{sabatini_aerodynamic_2024} proposed a linear model predictive control solution to the formation-flying problem, while simultaneously controlling the attitude by using a linear model of the differential dynamics. Similarly to Leonard, the idea of retractable drag plates was used. This simplification circumvent the necessity to model the variation in the drag effects as a function of the attitude.
\\

In 2015, \citeauthor{ben-yaacov_analytical_2015} derived an analytical algorithm for calculating the effective surface area of non-convex satellites based on theory of convex polygons. Using this algorithm, \cite{mishne_collision-avoidance_2017} devised a collision avoidance scheme for propulsion-less satellites, leveraging the change in drag and solar radiation pressure that occurs by reorienting the satellite. This algorithm outputs the required attitude to perform the desired maneuver to the attitude determination and control (ADCS), and therefore does not directly control $S_{\mathrm{eff}}$.
\\


\subsection{Contributions}
In this paper we present a kinematic framework which allows for the design of feedback control laws to \emph{directly} track the desired effective surface area of non-convex satellites, as opposed to indirect tracking via attitude control. This approach leaves a rotational degree of freedom, which can be used to perform other operations without impeding trajectory tracking of the effective surface area. In this work, we present a solution for maximizing the solar panel exposure as an example.



\section{Preliminaries}

\subsection{Notations and Coordinate Frames}
 We denote the frame $a$ by $\{a\}$. $\boldsymbol{b}^a$ denotes the vector $\boldsymbol{b}$ referenced in frame $a$. $^a\frac{d}{dt} \boldsymbol{b}$ is the time derivative of $\boldsymbol{b}$ differentiated within $a$. The skew-symmetric matrix operator is denoted by $\boldsymbol{S}(\cdot)$, such that for two arbitrary vectors $\boldsymbol{a}$ and $\boldsymbol{b}$, $\boldsymbol{S}(\boldsymbol{a}) \boldsymbol{b} = \boldsymbol{a} \times \boldsymbol{b}$. $\boldsymbol{A}^\top$ denotes the transpose of $\boldsymbol{A}$, while $\boldsymbol{A}^\dagger$ represents its Moore-Penrose inverse.  Furthermore $\dot{\boldsymbol{b}} := \frac{d}{dt} \boldsymbol{b}$, and $\Vert \boldsymbol{b} \Vert := \sqrt{\boldsymbol{b} \cdot \boldsymbol{b}} $. $\boldsymbol{a}\parallel \boldsymbol{b}$ signify that $\boldsymbol{a}$ and $\boldsymbol{b}$ are parallel.

To develop the system dynamics and kinematics, we rely on the following coordinate frame definitions:

\textit{Geocentric Inertial Frame $\{i\}$:} This frame has its origin in the center of the Earth, with its $z$-axis pointed towards the North Pole, while the $x$-axis is pointed towards the vernal equinox. 

\textit{Body Frame $\{b\}$:} This frame is has its origin in the geometrical center of the satellite, with its axes along the principal axes of inertia. 

\subsection{Dynamics of the Pseudoinverse}
For a matrix $\boldsymbol{A} \in \mathbb{R}^{m \times n}$, the derivative of its Moore-Penrose inverse $\frac{d}{dt}(\boldsymbol{A}^\dagger)$ is given by \citep{magnus_matrix_2019}
\begin{multline} \label{ch:math:eq:derivativePseudoInvG}
\frac{d}{dt}(\boldsymbol{A}^{\dagger}) =  - \boldsymbol{A}^\dagger \dot{\boldsymbol{A}}\boldsymbol{A}^\dagger + \boldsymbol{A}^\dagger \left(\boldsymbol{A}^\dagger\right)^\top \dot{\boldsymbol{A}}^\top \left( \boldsymbol{I}_m - \boldsymbol{A} \boldsymbol{A}^\dagger\right) \\ + \left( \boldsymbol{I}_n - \boldsymbol{A}^\dagger \boldsymbol{A} \right) \dot{\boldsymbol{A}}^\top \left(\left( \boldsymbol{A}^\dagger\right)^\top \boldsymbol{A}^\dagger\right).
\end{multline}

\subsection{Quaternion Attitude Kinematics and Dynamics}
The attitude kinematics describe how the angular velocities in $\{b\}$ relate to the change in attitude of $\{b\}$ relative to $\{i\}$. Using the quaternion notation, we may write \citep{egeland_modeling_2003} 
\begin{align} \label{eq:quatkin}
\dot{\mathbf{q}} &= \frac{1}{2} \begin{bmatrix}
    -\boldsymbol{\epsilon}^\top \\
    \eta \mathbf{I_3} + \boldsymbol{S}(\boldsymbol{\epsilon})
\end{bmatrix} \boldsymbol{\omega} \\ \label{eq:Attitude_kin}
    \dot{\mathbf{q}} &:= \boldsymbol{T}(\boldsymbol{q})\boldsymbol{\omega}
\end{align}
where $\boldsymbol{T}(\boldsymbol{q}) \in \mathbb{R}^{4\times3}$ and $\boldsymbol{q}$ is the unit quaternion, representing the rotation of $\{b\}$ relative to $\{i\}$. Its real and imaginary parts are denoted by $\eta$ and $\boldsymbol{\epsilon}$, and $\boldsymbol{\omega}$ represents the angular velocity of $\{b\}$ relative to $\{i\}$, referenced in $\{b\}$.

The attitude dynamics can be represented as a moment balance about the center of mass of the satellite, such that \citep{egeland_modeling_2003}

\begin{equation} \label{eq:attitude_dyn_simplest}
    \boldsymbol{I} \dot{\boldsymbol{\omega}} - \boldsymbol{S}(\boldsymbol{I} \boldsymbol{\omega}) \boldsymbol{\omega} = \boldsymbol{\tau},
\end{equation}
where $\boldsymbol{I}$ is the inertia matrix, and \(\boldsymbol{\tau}= \begin{bmatrix}
    \tau_x&\tau_y&\tau_z
\end{bmatrix}^\top \) represents the control torques. Note that this simplified model does not contain any disturbance torques, as the focus of this work lies on the the general problem of tracking the effective surface area. 

\section{The Effective Surface Area}
For non-convex satellites, a general closed-form expression for the effective surface area $S_{\mathrm{eff}}$ does not exist. Hence, in our approach, $S_{\mathrm{eff}}$ is calculated with the algorithm presented by \cite{ben-yaacov_analytical_2015}. This algorithm calculates $S_{\mathrm{eff}}$ analytically by subtracting the overlapping area of each surface from the total surface area. We now derive the properties of $S_{\mathrm{eff}}$ necessary in the controller design.

\subsection{Effective Surface Area Dynamics}
The derivative of the effective surface area is a function of the attitude $\boldsymbol{q}$, the satellite velocity vector $\boldsymbol{v}$, defined as the linear velocity of $\{b\}$ relative to $\{i\}$, referenced in $\{i\}$  and angular velocity $\boldsymbol{\omega}$, such that
\begin{equation}
    \dot{S}_{\mathrm{eff}} = h(\boldsymbol{q},\boldsymbol{v},\boldsymbol{\omega}) .
\end{equation}
 As $S_{\mathrm{eff}}$ must be found algorithmically, a closed form expression for its derivative is not available. The expression is therefore found numerically with a finite difference scheme in addition to the chain rule.
Using the chain rule together with (\ref{eq:quatkin}), the time-derivative of $S_{\mathrm{eff}}$ becomes 
\begin{align}
    \mathrm{\dot{S}_{eff}(\boldsymbol{q},\boldsymbol{v}, \boldsymbol{\omega})} &= \underbrace{\frac{\partial S_{\mathrm{eff}}}{\partial \boldsymbol{q}}}_{\nabla^\top S_{\mathrm{eff}}} \boldsymbol{\dot{q}} + \underbrace{\frac{\partial S_{\mathrm{eff}}}{\partial \boldsymbol{v}} \boldsymbol{\dot v}}_{\approx 0}\\ \label{ch:math:eq:dSeff_G}
    &= \underbrace{\frac{\partial S_{\mathrm{eff}}}{\partial \boldsymbol{q}} \boldsymbol{T}(\boldsymbol{q})}_{\boldsymbol{G}} \boldsymbol{\omega}
\end{align}
where $\boldsymbol{G(q,v)} \in \mathbb{R}^{1\times3}$ is the effective surface area transformation vector that relates the angular velocity, to a change in $S_{\mathrm{eff}}$. For notational convenience, $\boldsymbol{G}(\boldsymbol{q}, \boldsymbol{v})$ is denoted as $\boldsymbol{G}$, when $\boldsymbol{q}$ and $\boldsymbol{v}$ are clear from the context. The gradient $\nabla S_{\mathrm{eff}} \in \mathbb{R}^4$ follows the column vector convention. The orbital velocity vector $\boldsymbol{v}$ is assumed to be slowly varying compared to $\boldsymbol{q}$; thus,  $\dot{\boldsymbol{v}} \approx \boldsymbol{0}$.
Each element $i$ of the gradient $\nabla S_{\mathrm{eff}}$ is found with a finite difference scheme.
\\

\subsection[Dynamics of \textbf{G}]{Dynamics of $\boldsymbol{G}$}
The time-derivative of $\boldsymbol{G}$ can be expressed as
\begin{equation}
    \dot{\boldsymbol{G}} =  \frac{d}{dt} \left( \frac{\partial S_{\mathrm{eff}}}{\partial \boldsymbol{q}}\right) \; \boldsymbol{T}(\boldsymbol{q})+ \frac{\partial S_{\mathrm{eff}}}{\partial \boldsymbol{q}} \frac{d}{dt} \boldsymbol{T(q)}
\end{equation}

and moreover, the term $\frac{d}{dt}\frac{\partial S_{\mathrm{eff}}}{\partial \boldsymbol{q}}$ can be written as
\begin{equation}
    \frac{d}{dt}\nabla^\top S_{\mathrm{eff}} = 
    \left(\frac{\partial^2 S_{\mathrm{eff}}}{\partial \boldsymbol{q}^2} \boldsymbol{\dot q} \right)^\top = \left( \frac{\partial^2 S_{\mathrm{eff}}}{\partial \boldsymbol{q}^2} \boldsymbol{T}(\boldsymbol{q}) \boldsymbol{\omega} \right)^\top.
\end{equation}

The dynamics of $\boldsymbol{G}$ is then given by
\begin{equation}
    \dot{\boldsymbol{G}} = \left(\frac{\partial^2 S_{\mathrm{eff}}}{\partial \boldsymbol{q}^2} \boldsymbol{T}(\boldsymbol{q}) \boldsymbol{\omega}\right)^\top \boldsymbol{T}(\boldsymbol{q})+ \frac{\partial S_{\mathrm{eff}}}{\partial \boldsymbol{q}} \frac{\partial \boldsymbol{T(q)}}{\partial \boldsymbol{q}}  \dot{\boldsymbol{q}},
\end{equation}

where $\frac{\partial^2 \mathrm{ S_{eff}}}{\partial \boldsymbol{q}^2}$ is the Hessian of $S_{\mathrm{eff}}$, and $\frac{\partial \boldsymbol{T} (\boldsymbol{q})}{\partial \boldsymbol{q}}$ is a third order tensor. The product $\frac{\partial \boldsymbol{T}(\boldsymbol{q}) }{\partial \boldsymbol{q}} \dot{\boldsymbol{q}}$ is further calculated as 
\begin{equation}
    \frac{\partial \boldsymbol{T}(\boldsymbol{q}) }{\partial \boldsymbol{q}} \dot{\boldsymbol{q}} = \sum_i \left( \frac{\partial \boldsymbol{T}(\boldsymbol{q}) }{\partial \boldsymbol{q}} \right)_i \dot{\boldsymbol{q}}_i.
\end{equation}

The Hessian of $S_{\mathrm{eff}}$ is found with a second-order finite-difference scheme \citep{nocedal_numerical_2006}.

%

\subsection{Lipschitz Continuity}\label{ch:ESA:sec:Lipschitz}
An important prerequisite for many Lyapunov-based control methods is that the system dynamics are locally Lipschitz. That is, for all pairs of quaternions $(\boldsymbol{q}_1, \boldsymbol{q}_2)$ there exists a constant $K>0$ such that
\begin{equation}
    \frac{\Vert \boldsymbol{G}(\boldsymbol{q}_1) - \boldsymbol{G}(\boldsymbol{q}_2)\Vert}{ \Vert \boldsymbol{q}_1 - \boldsymbol{q}_2\Vert} \leq K.
\end{equation}
This is however not a reasonable assumption for satellite geometries modeled as a polyhedron (3D shape consisting of polygon faces). These shapes have sharp edges, resulting in a instantaneous change in $\mathrm{\dot S_{eff}}$, when a face comes in, or goes out of projection. This in turn means that the derivative of $S_{\mathrm{eff}}$ is in fact \emph{not} locally Lipschitz on the whole state space. We must therefore make the following assumption.

\begin{assumption} \label{ch:ctrl:ass:GpiecewiseLipschitz}
    $\boldsymbol{G}$ is piecewise locally Lipschitz in $\boldsymbol{q}$. Specifically, there exists a finite collection of domains $\mathbb{D}_i \in \mathbb{D} \subset \mathbb{R}^4, i \in \mathcal{I}$ such that 
    \begin{equation}
        \boldsymbol{G} = \boldsymbol{G}_i \quad \forall \;\boldsymbol{q} \in \mathbb{D}_i,
    \end{equation}
    where $\boldsymbol{G}_i$ is locally Lipschitz in $\boldsymbol{q}$ on its respective domain $\mathbb{D}_i$.
\end{assumption}

With Assumption \ref{ch:ctrl:ass:GpiecewiseLipschitz}, we can consider $\mathrm{\dot S_{eff}}$ in (\ref{ch:math:eq:dSeff_G}) a nonlinear switched system, where the mode $\boldsymbol{G}_i$ switches when it crosses a switching surface $\mathcal{S}$. Two basic switching scenarios are illustrated in Figures \ref{ch:ctrl:fig:switching_cross} and \ref{ch:ctrl:fig:slding_mode}. Figure \ref{ch:ctrl:fig:switching_cross} depicts desirable behavior when the state crosses a switching surface, whereas Figure \ref{ch:ctrl:fig:slding_mode}, showcases the scenario where $\mathcal{S}$ is attractive to both modes $f_1$ and $f_2$. This causes rapid switching between the modes as the state trajectory slides along $\mathcal{S}$. The result is a sliding mode, leading to chattering. This can degrade performance, as the resulting system dynamics can greatly differ from both $f_1$, and $f_2$ \citep{liberzon_switching_2003}.

\begin{figure}[h]
    \centering
    \begin{subfigure}[t]{0.24\textwidth}
        \centering
        \includegraphics[width=\linewidth]{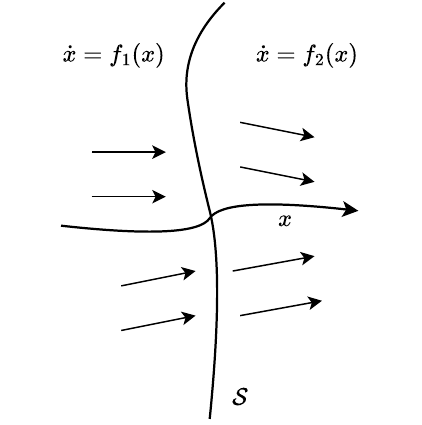}
        \caption{Switching surface}
        \label{ch:ctrl:fig:switching_cross}
    \end{subfigure}
    \hfill
    \begin{subfigure}[t]{0.24\textwidth}
        \centering
        \includegraphics[width=\linewidth]{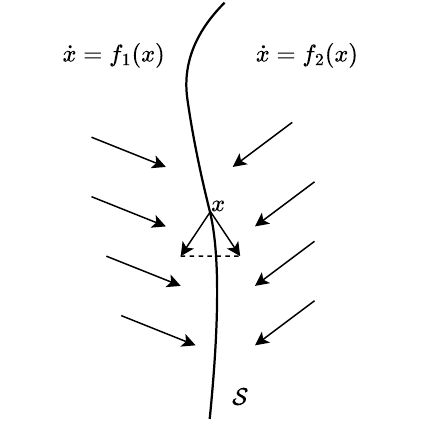}
        \caption{Sliding mode}
        \label{ch:ctrl:fig:slding_mode}
    \end{subfigure}
    \caption{Two cases of state behavior as the trajectory crosses $\mathcal{S}$ \citep{liberzon_switching_2003}}
    \label{ch:ctrl:fig:Switch_and_Slide}
\end{figure}

\subsection{Invariance}
A useful property of $\boldsymbol{G}$, that will allow for performing both trajectory tracking of $S_{\mathrm{eff}}$ while simultaneously performing other maneuvers using controller design based purely on feedback, is that any rotation about the velocity vector $\boldsymbol{v}^b$ does not affect $S_{\mathrm{eff}}$. To prove this we utilize the following Lemma \citep{aliprantis_principles_2007}: 

\begin{lemma} \label{ch:ap:lemma:lebesqueLinear}
    Let $\lambda(E)$ be the Lebesque measure of the set $E \subset \mathbb{R}^k$. If $A:\mathbb{R}^k \to \mathbb{R}^k$ is a linear operator, then 
    \begin{equation}
        \lambda (A(E)) = \vert\det A \vert \;
        \lambda(E)
    \end{equation}
    holds for all Lebesque measurable subsets $E$ of $\mathbb{R}^k$.
\end{lemma}

Based on this, we propose the following:
\begin{proposition} \label{ch:ctrl:prop:Ginvariant}
    The effective surface area transformation vector $\boldsymbol{G}$ is orthogonal to $\boldsymbol{v}^b$. This implies that 
    \begin{equation}
        \boldsymbol{\omega} \parallel \boldsymbol{v}^b \implies \dot{\mathrm{S}}_{\mathrm{eff}} = \boldsymbol{G} \boldsymbol{\omega} = 0.
    \end{equation}
\end{proposition}
\begin{pf}
    Without loss of generality, define a coordinate system with its $z$-axis pointed along $\boldsymbol{v}^b$. Let $p_{i,proj}(x,y)$ be the projection of each point, $p_i(x,y,z)$ of the satellite on the $xy$-plane. Let $S \subset \mathbb{R}^2$ be the set of all points $p_{i,proj}(x,y)$. Its Lebesque measure $\lambda:S \to \mathbb{R}$ is the area given by $\int_S dS$.

    From Lemma \ref{ch:ap:lemma:lebesqueLinear} we have
    \begin{equation}
        \lambda (\boldsymbol{\boldsymbol{R}}(S)) = \vert \det(\boldsymbol{R})\vert \;\lambda(S) = \lambda(S),
    \end{equation}
where $\boldsymbol{R} \in \mathbb{R}^{2\times2}$ is a rotation matrix. Thus a rotation about $\boldsymbol{v}^b$ does not alter $S_{\mathrm{eff}}$. It follows that projection of $\boldsymbol{v}^b$ onto $\boldsymbol{G}$ is zero, i.e. $\boldsymbol{v}^b \cdot \boldsymbol{G} = 0$.
\end{pf}

\section{Control Design}
We are now ready to state our main results. First, we present a backstepping-based controller for tracking the desired $S_{\mathrm{eff}}$. Second, we exemplify the advantages of the proposed approach by extending the controller to maximize the solar panel exposure in addition to tracking the desired $S_{\mathrm{eff}}$.
\subsection{System Equations}

We define the effective surface area error state $e_e = S_{\mathrm{eff,d}} - S_{\mathrm{eff}}$, where $S_{\mathrm{eff,d}} \in[S_{\mathrm{eff,min}}, S_{\mathrm{eff,max}}]$ is a bounded reference signal, with all its derivatives smooth and bounded. Recalling  \eqref{ch:math:eq:dSeff_G}, the error dynamics are
\begin{equation}
    \dot{e}_e  = \dot{S}_{\mathrm{eff,d}} - \boldsymbol{G} \boldsymbol{\omega}.
\end{equation}
The attitude, parametrized by $\boldsymbol{q}$, is assumed to be measured. The system equations are therefore given by the equation set
\begin{subequations} \label{ch:ctrl:eq:systemEqs}
\begin{align}
    \label{ch:ctrl:eq:seffdot}
    \dot{e}_e &= \dot{S}_{\mathrm{eff,d}} - \boldsymbol{G} \boldsymbol{\omega} \\ \label{ch:ctrl:eq:omegadot}
    \dot{\boldsymbol{\omega}} &= \boldsymbol{I}^{-1} \boldsymbol{S}(\boldsymbol{I} \boldsymbol{\omega}) \boldsymbol{\omega} + \boldsymbol{I}^{-1} \boldsymbol{\tau}.
\end{align}
\end{subequations}


\subsection{Trajectory Tracking}\label{ch:ctrl:sec:traj_track}
\begin{proposition}
 The control law
 \begin{equation}
    \boldsymbol{\tau} = \boldsymbol{I} \dot{\boldsymbol{\xi}} - \boldsymbol{S}(\boldsymbol{I \boldsymbol{\omega}})  - \boldsymbol{K}_z \boldsymbol{z}_1 + \boldsymbol{G}^\top e_e,
\end{equation}
where
\begin{equation}
    \dot{\boldsymbol{\xi}} = \boldsymbol{G}^\dagger \left( \ddot{S}_{\mathrm{eff,d}} + k_e \dot e_e \right) + \dot{\boldsymbol{G}^\dagger} \left(  \dot{S}_{\mathrm{eff,d}} + k_e e_e\right),
\end{equation}
\begin{equation}
    \boldsymbol{z}_1 =\boldsymbol{\omega} - \boldsymbol{G}^\dagger \left( \dot{S}_{\mathrm{eff,d}} + k_e e_e\right),
\end{equation}
and $\boldsymbol{K}_z \in \mathbb{R}^{3 \times 3}$, $k_e \in \mathbb{R}$ are control gains, leaves the equilibrium $e_e = 0$ of \eqref{ch:ctrl:eq:systemEqs} asymptotically stable.

\end{proposition}
\begin{pf}
We begin by choosing the first control Lyapunov function (CLF) as
\begin{equation}
    V_1 = \frac{1}{2} e_e^2
\end{equation}
and from inserting \eqref{ch:ctrl:eq:seffdot} with $\boldsymbol{\omega} = \boldsymbol{\xi} + \boldsymbol{z}_1$, we obtain
\begin{equation}
    \dot V_1 =  e_e (\dot{S}_{\mathrm{eff,d}} - \boldsymbol{G} (\boldsymbol{\xi} + \boldsymbol{z}_1)).
\end{equation}
Moreover, the stabilizing function $\boldsymbol{\xi}$ is chosen as
\begin{equation}
   \boldsymbol{\xi} =  \boldsymbol{G}^\dagger \left( \dot{S}_{\mathrm{eff,d}} + k_e e_e\right).
\end{equation}
Note that the stabilizing function is a function of $\boldsymbol{G}^\dagger$, which grows unbounded when $S_{\mathrm{eff}}$ approaches its maximum value. When $\boldsymbol{z}_1 = \boldsymbol{0}$ we have that
\begin{equation}
    \dot V_1 = -k_e e_e^2
\end{equation}
and hence, the unforced system \eqref{ch:ctrl:eq:seffdot}, with $\boldsymbol{z}_1$ viewed as the input is exponentially stable for all non-singular $\boldsymbol{G}$.
\\

We choose the composite CLF to be 
\begin{equation} \label{ch:ctrl:traj_track_V2}
    V_2 = V_1 + \boldsymbol{z}_1^\top \boldsymbol{I} \boldsymbol{z}_1,
\end{equation}
then its derivative is 
\begin{equation}
    \dot V_2 = -k_e e_e^2 - e_e \boldsymbol{G} \boldsymbol{z}_1 + \boldsymbol{z}_1^\top \left( \boldsymbol{S}(\boldsymbol{I \omega})\boldsymbol{\omega} + \boldsymbol{\tau} - \boldsymbol{I} \dot{\boldsymbol{\xi}} \right).
\end{equation}

Choose then the input as 
\begin{equation} \label{ch:ctrl:eq:tracjTrac:tau}
    \boldsymbol{\tau} = \boldsymbol{I} \dot{\boldsymbol{\xi}} - \boldsymbol{S}(\boldsymbol{I \boldsymbol{\omega}})  - \boldsymbol{K}_z \boldsymbol{z}_1 + \boldsymbol{G}^\top e_e,
\end{equation}
where
\begin{equation}
    \dot{\boldsymbol{\xi}} = \boldsymbol{G}^\dagger \left( \ddot{S}_{\mathrm{eff,d}} + k_e \dot e_e \right) + \dot{\boldsymbol{G}^\dagger} \left(  \dot{S}_{\mathrm{eff,d}} + k_e e_e\right)
\end{equation}
and $\dot{\boldsymbol{G}^\dagger}$ is calculated according to \eqref{ch:math:eq:derivativePseudoInvG}.

Hence, 
\begin{equation}
    \dot V_2 = - k_e e_e^2 - \boldsymbol{z}^\top_1 \boldsymbol{K}_z \boldsymbol{z}_1 < 0 \; \forall \{e_e,\boldsymbol{z}_1, S_{\mathrm{eff}}\} \neq \{0,\boldsymbol{0}, \mathrm{S_{eff,max}}\}.
\end{equation}
Thus, $V_2$ is a strict Lyapunov function for $e_e$ and $\boldsymbol{z}_1$. It is then concluded that all trajectories for which  $S_{\mathrm{eff}} \neq \mathrm{S_{eff,max}}$, approach the equilibrium $\{e_e,e_s\} = \{0,0\}$. As a result of the rank-deficiency of $\boldsymbol{G}$ when $S_{\mathrm{eff}} = \mathrm{S_{eff,max}}$, we can only conclude the equilibrium to be exponentially stable. This is a direct consequence of $\boldsymbol{G}^\dagger$ growing unbounded for this state, and consequently the stabilizing function, $\boldsymbol{\xi}$ grows unbounded for any nonzero $e_e$ and $\dot{S}_{\mathrm{eff,d}}$.

Returning to the untransformed system, driving $\boldsymbol{z}_1$ to zero, implies that 
\begin{equation}
     \boldsymbol{z}_1 = \boldsymbol{\omega} - \boldsymbol{G}^\dagger(\dot{S}_{\mathrm{eff,d}} + k_e e_e) = 0.
\end{equation}
Pre-multiplying with $\boldsymbol{G}$ yields
\begin{equation}
    \underbrace{\boldsymbol{G \omega} - \dot{S}_{\mathrm{eff,d}}}_{-\dot e_e}  - k_e e_e = 0.
\end{equation}
Hence,
\begin{equation}
    \dot e_e = - k_e e_e,
\end{equation}
and for $e_e$ = 0,
\begin{equation}
    \boldsymbol{G\omega} = \dot{S}_{\mathrm{eff,d}}.
\end{equation}
Driving $\boldsymbol{z}_1$ to the origin implies that $e_e$ converges exponentially to zero, and for $e_e = 0$, the input compensates for the change in the reference.
\end{pf}
\subsection{Sun pointing in the invariant subspace of $S_{\mathrm{eff}}$}
Let $\boldsymbol{v}^b$ be the velocity vector, $\boldsymbol{n}^b$ be the solar panel normal vector and $\boldsymbol{k}^b$ be the vector pointing from the satellite to the Sun. These vectors are unit vectors, expressed in $\{b\}$. Moreover, $\boldsymbol{v}^i$, and $\boldsymbol{k}^i$ are both assumed to be slowly varying in $\{i\}$.
\\

The objective is to maximize the alignment of $\boldsymbol{k}^b$ with $\boldsymbol{n}^b$, without affecting $S_{\mathrm{eff}}$. Following Proposition \ref{ch:ctrl:prop:Ginvariant}$, S_{\mathrm{eff}}$ is invariant to a rotation about $\boldsymbol{v}^b$. Thus we wish to reformulate the problem such that the resultant angular velocity is parallel to $\boldsymbol{v}^b$.
\\

The first step is to define the vector $\boldsymbol{k}^b_\perp$, which is co-planar with both $\boldsymbol{v}^b$ and $\boldsymbol{k}^b$ while simultaneously perpendicular to $\boldsymbol{v}_b$. Mathematically, this can be written as

\begin{equation}
    \boldsymbol{k}^{b}_{\perp} = \frac{(\boldsymbol{v}^b \times \boldsymbol{k}^b) \times \boldsymbol{v}^b}{\Vert(\boldsymbol{v}^b \times \boldsymbol{k}^b) \times \boldsymbol{v}^b \Vert}.
\end{equation}

Similarly for $\boldsymbol{n}^b_\perp$, which lies in the plane defined by $\boldsymbol{v}^b$ and $\boldsymbol{n}^b$, while perpendicular to $\boldsymbol{v}^b$, we have that
\begin{equation}
    \boldsymbol{n}^b_{\perp} = \frac{(\boldsymbol{v}^b \times \boldsymbol{n}^b) \times \boldsymbol{v}^b}{\Vert (\boldsymbol{v}^b \times \boldsymbol{n}^b) \times \boldsymbol{v}^b \Vert}.
\end{equation}

Writing $\boldsymbol{k}^b_\perp$ in terms of $\{i\}$,
\begin{equation}
    \boldsymbol{k}^b_\perp = \boldsymbol{R}^b_i k^i_\perp,
\end{equation}
we can formulate its derivative as
\begin{equation}
    \dot{\boldsymbol{k}}^b_\perp  = \dot{\boldsymbol{R}^b_i} \boldsymbol{k}^i_\perp
    = - \boldsymbol{S}(\boldsymbol{\omega}) \boldsymbol{k}^b_\perp.
\end{equation}
The Sun pointing error state is defined as
\begin{equation}
e_s = 1 - \boldsymbol{k}^b_\perp \cdot \boldsymbol{n}^b_\perp.
\end{equation}

Additionally, we define a new rotating reference frame, in which $\boldsymbol{n}_b$ is inertial, as follows:
\begin{defn} \label{ch:ctrl:def:s_coord} 
Let $\{s\}$ be the body-fixed rotating reference frame given by
$\boldsymbol{i}_s = \boldsymbol{v}^b$, $\boldsymbol{j}_s = \boldsymbol{n}^b_\perp$, $\boldsymbol{k}_s$ = $\boldsymbol{i}_s \times \boldsymbol{j}_s$.
\end{defn}
We wish to find the error dynamics in the rotating reference frame $\{s\}$, and begin by defining
\begin{equation}
    \boldsymbol{k}^s_\perp = \boldsymbol{R}^s_i \boldsymbol{k}^i_\perp.
\end{equation}
Its dynamics then becomes
\begin{equation}
    \dot{\boldsymbol{k}}^s_\perp =  -\boldsymbol{S}(\boldsymbol{\omega}_s)\boldsymbol{k}^s_\perp,
\end{equation}
where $\boldsymbol{\omega}_s$ is the angular velocity of $\{s\}$ relative to $\{i\}$ expressed in $\{s\}$. The error dynamics differentiated in $\{s\}$ becomes
\begin{align}
    \frac{^sd}{dt} e_s&= -\frac{^sd}{dt} \boldsymbol{k}^s_\perp \cdot \boldsymbol{n}^s_\perp \\
    &= (\boldsymbol{k}^s_\perp)^\top \boldsymbol{S}(\boldsymbol{n}^s_\perp) \boldsymbol{\omega}^s \\
    \implies  \frac{^sd}{dt} e_s &= \underbrace{(\boldsymbol{k}^b_\perp)^\top \boldsymbol{S}(\boldsymbol{n}^b_\perp)}_{\boldsymbol{H}} \boldsymbol{\omega}. \label{eq:e_s_dot}
\end{align}
Note that $\boldsymbol{H}$ is the cross product of $\boldsymbol{k}^b_\perp$ and $\boldsymbol{n}^b_\perp$. When these vectors become linearly dependent, $\Vert \boldsymbol{H} \Vert = 0$. This corresponds to the error states $e_s = 0 \lor 2$.

\begin{corollary} \label{ch:ctrl:cor:H_G_orthogonal}
From Proposition \ref{ch:ctrl:prop:Ginvariant} it follows directly that
    \begin{equation}
        \boldsymbol{GH}^\top \equiv 0,
    \end{equation}
    and consequently,
    \begin{equation}
        \boldsymbol{HG}^\dagger \equiv 0. 
    \end{equation}
\end{corollary}
\begin{pf}
    $\boldsymbol{H}$ is parallel to $\boldsymbol{v}^b$. From Proposition \ref{ch:ctrl:prop:Ginvariant}, $ \boldsymbol{G}$ is orthogonal to $\boldsymbol{v}_b$, and it then follows that 
    \begin{equation}
         \boldsymbol{GH}^\top \equiv 0.
    \end{equation}
    From the definition of the pseudoinverse, 
    \begin{equation}
        \boldsymbol{G}^\dagger \parallel \boldsymbol{G}^\top,
    \end{equation}
    which implies that
    \begin{equation}
         \boldsymbol{HG}^\dagger \equiv 0.
    \end{equation}
\end{pf}


\begin{theorem}
    The control law
    \begin{equation}
        \boldsymbol{\tau} = -\boldsymbol{S}(\boldsymbol{I \omega})\boldsymbol{\omega} + \boldsymbol{I} \dot{\boldsymbol{\xi}} - \boldsymbol{K}_z \boldsymbol{z}_1 + (k_s e_s \boldsymbol{H} + k_e e_e \boldsymbol{G})^\top,
    \end{equation}
    leaves the equilibrium $e_e = 0$ of \eqref{ch:ctrl:eq:systemEqs} and $e_s = 0$ of \eqref{eq:e_s_dot} asymptotically stable.
\end{theorem}
\begin{pf}
Choose the CLF as
\begin{equation}
    V_1 = \frac{1}{2} k_s e_{s}^2 +  \frac{1}{2} k_e e_e^2.
\end{equation}
Its derivative is then
\begin{equation} \label{ch:ctrl:eq:componentBackstep:V1}
    \dot V_1 = - k_s e_s \boldsymbol{H} \left( \boldsymbol{z}_1 + \boldsymbol{\xi} \right) + k_e e_e \left( \dot{S}_{\mathrm{eff,d}} -  \boldsymbol{G} \left( \boldsymbol{z}_1 + \boldsymbol{\xi} \right) \right),
\end{equation}
where $\boldsymbol{\omega} = \boldsymbol{z_1} + \boldsymbol{\xi}$.
Choose the stabilizing function as 
\begin{equation} \label{ch:ctrl:eq:componentBackstep:xi}
    \boldsymbol{\xi} = \boldsymbol{H}^\top k_s e_s + \boldsymbol{G}^\dagger \left( \dot{S}_{\mathrm{eff,d}} + k_e e_e\right).
\end{equation}
Inserting \eqref{ch:ctrl:eq:componentBackstep:xi} into \eqref{ch:ctrl:eq:componentBackstep:V1} yields
\begin{align}
\dot{V}_1 =& - k_s e_s \boldsymbol{H} \boldsymbol{H}^\top k_s e_s - k_s e_s \underbrace{\boldsymbol{H} \boldsymbol{G}^\dagger}_{0} \left( \dot{S}_{\mathrm{eff,d}} + k_e e_e\right) \notag\\ 
&- k_e e_e \boldsymbol{G} \boldsymbol{G}^\dagger k_e e_e 
    - k_e e_e \underbrace{\boldsymbol{G} \boldsymbol{H}^\top}_{0} k_s e_s\notag\\ 
&- \left( k_s e_s \boldsymbol{H} + k_e e_e \boldsymbol{G} \right) \boldsymbol{z}_1\\
=& - k_s^2 \boldsymbol{H H}^\top e_s^2 - k_e^2 e_e^2 - \left( k_s e_s \boldsymbol{H} + k_e e_e \boldsymbol{G} \right) \boldsymbol{z}_1
\end{align}
with $\boldsymbol{z}_1$ viewed as the input and set to zero, $\dot{V}_1<0 \;\; \forall \; \{e_s, e_e\} \neq \{0 \lor 2,0\}$. Here we used Corollary \ref{ch:ctrl:cor:H_G_orthogonal} to eliminate the cross-terms. 

The composite CLF is chosen as
\begin{equation} \label{ch:ctrl:eq:componentBackstep:V2}
    V_2 = V_1 + \boldsymbol{z}_1^\top \boldsymbol{I} \boldsymbol{z}_1
\end{equation}
from which we obtain the derivative
\begin{equation}
    \dot V_2 = \dot V_1 + \boldsymbol{z}_1^\top \left( \boldsymbol{S}(\boldsymbol{I \omega})\boldsymbol{\omega} + \boldsymbol{\tau} - \boldsymbol{I \dot{\boldsymbol{\xi}}} \right).
\end{equation}
Inserting the input torque
\begin{equation} \label{ch:ctrl:componentBacstep:tau}
    \boldsymbol{\tau} = -\boldsymbol{S}(\boldsymbol{I \omega})\boldsymbol{\omega} + \boldsymbol{I} \dot{\boldsymbol{\xi}} - \boldsymbol{K}_z \boldsymbol{z}_1 + (k_s e_s \boldsymbol{H} + k_e e_e \boldsymbol{G})^\top,
\end{equation}
yields the final derivative to be
\begin{equation}
    \dot V_2 = - k_s^2 \boldsymbol{H H}^\top e_s^2 - k_e^2 e_e^2 - \boldsymbol{z}_1^\top\boldsymbol{K}_z \boldsymbol{z}_1 < 0 
\end{equation}
for all ${\{e_e,e_s,S_{\mathrm{eff}}}\} \neq \{0, 0\lor 2, \mathrm{S_{eff,max}}\}$.
Thus \eqref{ch:ctrl:componentBacstep:tau} drives $e_s$, $e_e$ and $\boldsymbol{z}_1$ to zero, and the transformed system is asymptotically stable. From the definition of $\boldsymbol{z}_1$
\begin{align}
\boldsymbol{z}_1 = \boldsymbol{0} \implies \boldsymbol{\omega} &= \boldsymbol{\xi}, \\
&= \boldsymbol{G}^\dagger \dot{S}_{\mathrm{eff,d}}, \\
\implies \dot e_e &= 0.
\end{align}
Thus all trajectories of $e_e$ for which $\boldsymbol{G} \neq \boldsymbol{0}$ will converge exponentially to the origin, while all trajectories of $e_s$ for which $e_s \neq 2$ will converge asymptotically to the origin. Hence, $e_e$, $e_s$ and $\boldsymbol{z}_1$ are asymptotically stable.
\end{pf}

\section{Simulations}
Consider the CubeSat shown in Figure \ref{ch:math:fig:cubesatFig}, consisting of a $10 \: \mathrm{cm} \times 10 \: \mathrm{cm} \times 15 \: \mathrm{cm}$ hull with uniformly distributed mass $m_c = 2 \; \mathrm{kg}$, and two side-mounted $15 \mathrm{cm} \times 15 \mathrm{cm}$ solar panels, each with uniformly distributed mass $m_s = 0.5  \; \mathrm{kg}$.
\begin{figure}[h]
    \centering
    \includegraphics[width=0.6\linewidth]{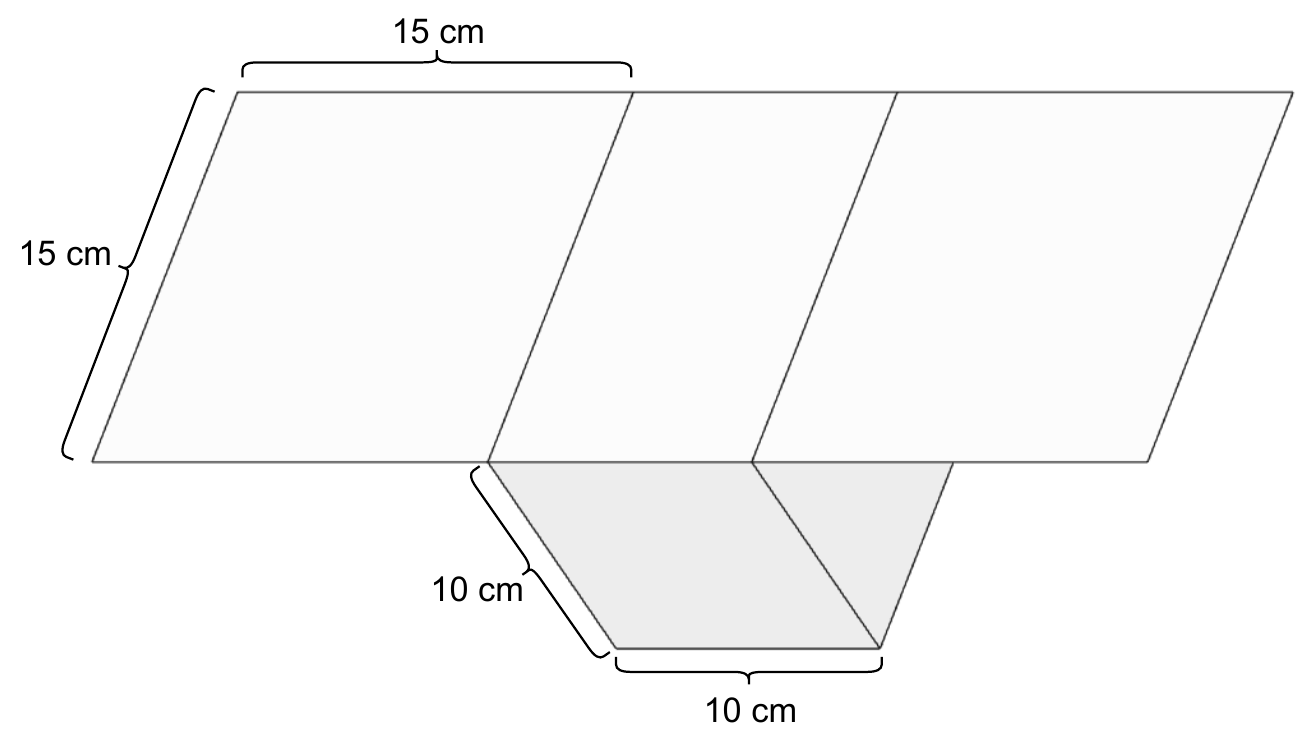}
    \caption{Dimensions of the CubeSat used in the simulations}
    \label{ch:math:fig:cubesatFig}
\end{figure}

The associated effective surface area is depicted in Figure \ref{fig:S_eff_anayltical}. The effective surface area is parametrized by the extrinsic Euler angles $\theta$ and $\psi$, corresponding to rotations about the $y$- and $z$-axes, respectively. To exemplify how $S_{\mathrm{eff}}$ evolves as a function of $\theta$ and $\psi$, the resulting trajectory from a simulation, later detailed in Figure \ref{fig:sim_sunProj} is superimposed. This figure highlights that $S_{\mathrm{eff}}$ not smooth everywhere, which results in discontinuities in its derivative. The minima occur when only the $10 \times 10$ \si{cm^2} is projected, while the maxima occur when both the solar panels, and the $10 \times 10$ \si{cm^2} surface is exposed. 

\begin{figure}[ht]
    \centering
    \includegraphics[width=0.8\linewidth]{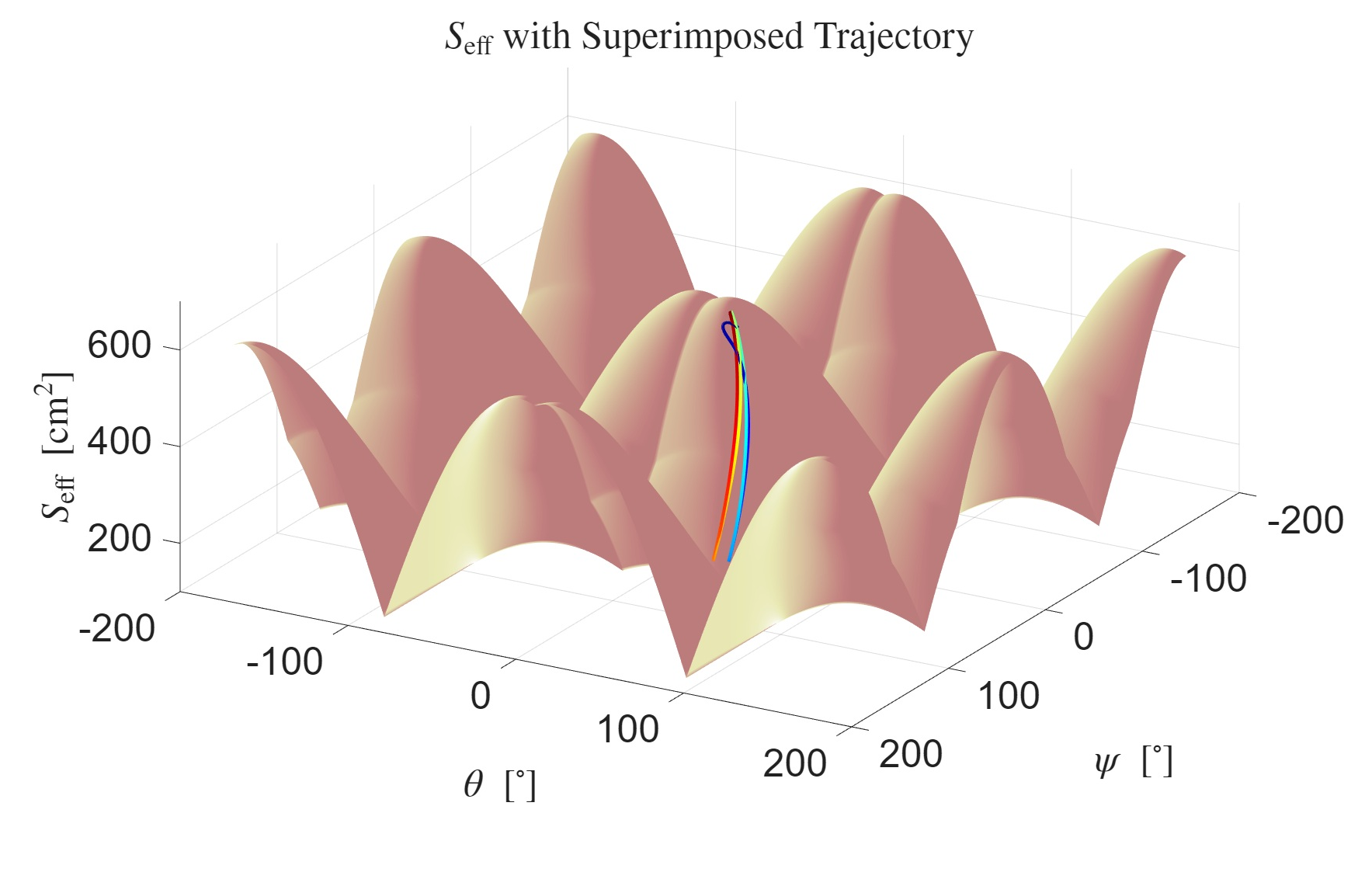}
    \caption{The effective surface area of the CubeSat as a function $\theta$ and $\psi$, superimposed with the resulting trajectory of the simulation in Figure \ref{fig:sim_sunProj}, with its color (blue $t = 0$, red $t = 40$) encoding the time axis. }
    \label{fig:S_eff_anayltical}
\end{figure}

The parameters and initial conditions used for gathering the simulation results are provided in Table \ref{tab:init_conds}.
\begin{table}[h]
    \centering
    \caption{Parameters and initial conditions}
    \begin{tabular}{lccl}
        Parameter & Symbol  & Value & Unit \\ \hline
        Quaternion & $\boldsymbol{q}$ & $10^{-1}$[9.6 1.4 2.0 1.4] & \\ 
       Angular rates & $\boldsymbol{\omega}$ & [0 0 0] & [rad/s]\\
         Step size & $h$ & 0.01 & [s]\\
         Inertia & $\boldsymbol{I}$ & $ 10^{-2} \: diag(2.48, 0.90, 2.25)$ & $\mathrm{kgm^2}$ \\
         Velocity vector & $\boldsymbol{v}^i$ & [1 0 0] &  \\
         Sun vector & $\boldsymbol{k}^i$ & $\frac{1}{\sqrt{2}}$[0 1 1] & \\
    \end{tabular}
    \label{tab:init_conds}
\end{table}

The satellite tracks the desired $S_{\mathrm{eff}}$, given as a sinusoidal trajectory. Figure \ref{fig:sim_trajTrack} shows the simulation in the case of pure trajectory tracking. The controller gains are chosen as $k_e = 10^{-6}$ and $\boldsymbol{K}_z = diag(0.05, 0.05, 0.05)$.  The desired trajectory spans $S_{\mathrm{eff,d}} = [100, 600] \; \mathrm{cm^2}$, with a period of 20 seconds. Here it is observed that the satellite closely follows the desired trajectory, except when $S_{\mathrm{eff}}$ goes below 200 $\mathrm{cm^2}$. Due to the switched nature of the $\boldsymbol{G}$, sliding modes can occur. These sliding modes cause chattering, degrading controller performance, as the system dynamics are greatly altered under rapid switching.

\begin{figure}
    \centering
    \includegraphics[width=0.9\linewidth]{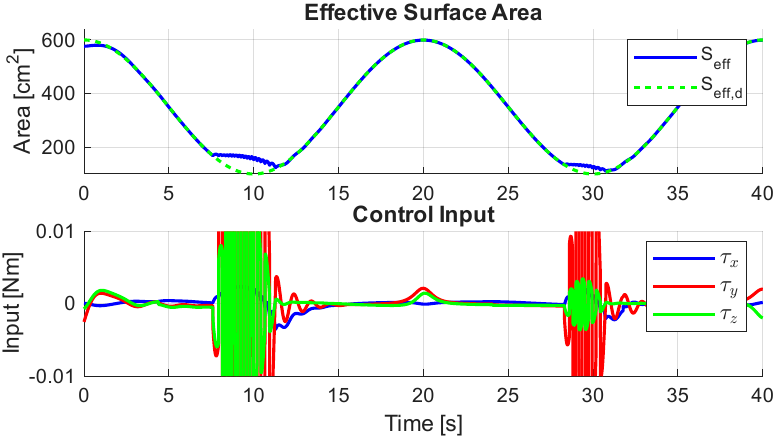}
    \caption{Simulation of the satellite tracking $S_{\mathrm{eff,d}}$}
    \label{fig:sim_trajTrack}
\end{figure}

Figure \ref{fig:sim_sunProj} shows the simulation in the case of trajectory tracking and Sun projection. Here the desired trajectory spans $S_{\mathrm{eff,d}} = [200, 600] \; \mathrm{cm^2}$, to avoid chattering. The gains are chosen as $k_e = 10^{-6}$, $k_s = 0.1$ and $\boldsymbol{K}_z = diag(0.03,0.03,0.03)$. The satellite closely tracks $S_{\mathrm{eff,d}}$, while maximizing the solar panel Sun projection. We notice that the sun projection reaches a minimum as $S_{\mathrm{eff}}$ reaches the peak of its trajectory, 20 seconds into the simulation. This occurs since the satellite must leverage the surface area of the solar panels to increase $S_{\mathrm{eff}}$, at the expense of the Sun projection. The resulting evolution in Figure \ref{fig:S_eff_anayltical}, which shows that $S_{\mathrm{eff}}$ smoothly moves along its manifold, further confirms the efficacy of the proposed approach. 


\begin{figure}
    \centering
    \includegraphics[width=0.9\linewidth]{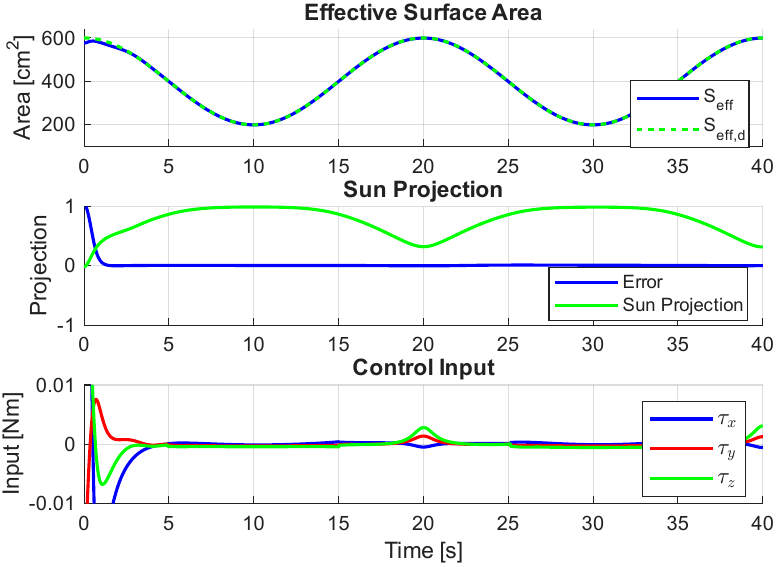}
    \caption{Simulation of the satellite tracking $S_{\mathrm{eff,d}}$ while performing Sun pointing}
    \label{fig:sim_sunProj}
\end{figure}

\section{Conclusions}
This paper presented a framework to perform tracking of the effective surface area, and  exemplified its advantages by simultaneously maximizing the Sun exposure of the satellite solar panels. Two controllers were presented, based on the backstepping design scheme. The controllers were, respectively, proven to be exponentially, and asymptotically stable. However the chattering that occurs when the state trajectory enters a sliding mode, can inhibit performance. Further research should focus on developing the outer loop, that generates the desired effective surface area. Additionally, limiting the impact of sliding modes, to reduce chattering should also be investigated.

\bibliography{references}             
                                                   








     
\end{document}